\documentclass[11pt]{article}
\usepackage{multirow}
\usepackage{booktabs} 

\usepackage{wrapfig}
\usepackage[]{caption}  
\usepackage{wrapfig}
\usepackage{graphicx, subfigure}

\newtheorem{assumption}{Assumption}
\newtheorem{problem}{Problem}
\newtheorem{remark}{Remark}
\newtheorem{lemma}{Lemma}
\newtheorem{definition}{Definition}
\newtheorem{corollary}{Corollary}
\newtheorem{theorem}{Theorem}
\newtheorem{proposition}{Proposition}
\usepackage{amsmath}
\usepackage{amsfonts}
\usepackage{amssymb}
\usepackage{cite}
\usepackage[capitalise]{cleveref}

\newcommand{\acks}[1]{\section*{Acknowledgments}#1}

\newenvironment{keywords}
{\bgroup\leftskip 20pt\rightskip 20pt \small\noindent{\bfseries
Keywords:} \ignorespaces}%
{\par\egroup\vskip 0.25ex}

\newenvironment{proof}{\paragraph{Proof:}}{\hfill$\square$}

\usepackage{enumitem}
\usepackage[margin=1in]{geometry}

\def\1{\mathbf{1}}
\newcommand{\MM}{\mathcal{M}}

\newcommand{\G}{G}
\newcommand{\E}{E}
\newcommand{\V}{V}

\newcommand{\LL}{\mathcal{L}}

\newcommand{\aij}{\alpha_{ij}}
\newcommand{\bij}{\beta_{ij}}

\newcommand{\GL}{\G_{\LL}}
\newcommand{\Ni}{\mathcal{N}_i}

\newcommand{\Nin}{\Ni^{\rm in}}
\newcommand{\Nout}{\Ni^{\rm out}}

\newcommand{\No}{\mathcal{N}^{\rm out}}

\newcommand{\mC}{\mathcal{C}}
\newcommand{\mD}{\mathcal{D}}

\newcommand{\mG}{\mathcal{G}}

\newcommand{\mL}{\mathcal{L}}
\newcommand{\mM}{\mathcal{M}}

\newcommand{\mV}{\mathcal{V}}

\newcommand{\GW}[1]{\G({#1})}
\newcommand{\VW}[1]{\V({#1})}
\newcommand{\EW}[1]{\E({#1})}

\newcommand{\ind}{\mathbf{1}_{\{q\in \LL\}}}
\newcommand{\errc}{\Delta_{\mC_q}}
\newcommand{\ol}{o_{\LL_q}}
\newcommand{\errq}{\Delta_{\mL_q}}
\newcommand{\infnorm}[1]{\norm{#1}_{\infty}}

\newcommand{\footremember}[2]{%
    \footnote{#2}
    \newcounter{#1}
    \setcounter{#1}{\value{footnote}}%
}
\newcommand{\footrecall}[1]{%
    \footnotemark[\value{#1}]%
}

\providecommand{\norm}[1]{\lVert#1\rVert}

\title{Learning Trust Over Directed Graphs in Multiagent Systems}
\usepackage{times}

\author{
  Orhan Eren Akg\"un\footremember{harvardseas}{ School of Engineering and Applied Sciences (SEAS), Harvard University}\\
  \texttt{erenakgun@g.harvard.edu}
  \and
  Arif Kerem Dayı\footrecall{harvardseas}\\
  \texttt{keremdayi@college.harvard.edu}
  \and 
  Stephanie Gil\footrecall{harvardseas}\\
  \texttt{sgil@seas.harvard.edu}
  \and 
  Angelia Nedi\'c\footremember{asueece}{ECEE, Arizona State University}\\
  \texttt{Angelia.Nedich@asu.edu}
}
\date{}

\begin{document}

\maketitle
\begin{abstract}%
We address the problem of learning the legitimacy of other agents in a multiagent network when an unknown subset is comprised of malicious actors. We specifically derive results for the case of directed graphs and where stochastic side information, or observations of trust, is available. We refer to this as ``learning trust'' since agents must identify which neighbors in the network are reliable, and we derive a protocol to achieve this. We also provide analytical results showing that under this protocol i) agents can learn the legitimacy of all other agents almost surely, and that ii) the opinions of the agents converge in mean to the true legitimacy of all other agents in the network. Lastly, we provide numerical studies showing that our convergence results hold in practice for various network topologies and variations in the number of malicious agents in the network. %
\end{abstract}


\section{Introduction}

Learning the network topology in multiagent systems, what edges exist and are reliable, is critical because of the central role it plays in many multiagent collaboration tasks. This includes a wide range of tasks from estimation, to control, to machine learning, optimization and beyond \cite{rabbat2004distributed,olshevsky2010efficient,nedic2018network}.  Many times both the coordination protocols and achievable performance of the team is dictated by topology~\cite{olfati2004consensus, xi2018linear, cai2012average, nedic2014distributed}. Two aspects that can greatly complicate the learning however, are i) directed graphs, and ii) the presence of untrustworthy data. Directed graphs are more common in practice due to heterogeneity in sensing and communication capabilities in multiagent systems, but are often more difficult to analyze due to non-symmetric information flow. On the other hand, the presence of malicious agents are an important real-world consideration but lead to untrustworthy data in the system~\cite{lamport2019byzantine,sundaram2010distributed,sundaram2018distributed,fischer1985impossibility}. Unfortunately, the compounded impact of both of these challenges is a very complex problem with sparse theory to date. \textbf{\emph{Our objective in this paper is to develop a learning protocol and its related analysis, where agents learn over time the legitimacy of their neighbors in the presence of malicious agents over directed graphs.}}

The class of problems over directed graphs pose a particular challenge to achieving resilience: many distributed algorithms on directed graphs require agents to have some information about their out-neighbors, but because of the asymmetric information flow, they cannot sense or obtain information directly from these agents. This makes detection of malicious out-neighbors particularly difficult. For instance, the distributed optimization algorithms presented in \cite{nedic2014distributed, tsianos2012push, tsianos2012consensus, makhdoumi2015graph, push-pull2021} and the distributed consensus algorithms \cite{cai2012average,dominguez2012distributed} all require that the agents know the number of out-neighbors they have. This assumption can break if an agent designs the update rule considering an out-neighbor as legitimate, but that agent is malicious in reality. Hence, agents need to have some information about the trustworthiness of their out-neighbors. An interesting concept that has the potential to help this difficult problem is the use of``side information" or data in cyberphysical systems \cite{liu2019trust,xiong2013securearray, pasqualetti2015control,renganathan2017spoof,gil2017guaranteeing,LCSS,giraldo2018survey,mallmann2021crowd, cavorsi2022exploiting}. Recent work has shown that by leveraging physical channels of information in the system, agents can gain stochastic information about the trustworthiness of the other agents~\cite{liu2019trust,giraldo2018survey,xiong2013securearray,gil2017guaranteeing}. We call these ``stochastic observations of trust.'' It has been shown that exploiting these observations leads to stronger results in resilience for multiagent systems~\cite{ourTRO,LCSS,mallmann2021crowd}. Unfortunately however, existing results do not immediately extend to the case of directed graphs. 

In this work, we are interested in learning a trusted graph topology over a directed graph. Using stochastic information about trustworthy neighbors, agents can decide how they should process information that they receive from their in-neighbors, and with which out-neighbors they should share their information. Since agents cannot necessarily observe their out-neighbors, it is natural to think that they need to get information about their out-neighbors from the other agents. We investigate what sufficient information agents can share and how they should process this information to learn the trustworthiness of the other agents in the system in a robust way. This setup is particularly challenging since there might be malicious agents in the system sharing misinformation during this learning process. We present a learning protocol to enable each agent to learn the trustworthiness of all other agents in the system leveraging the opinion of their neighbors. Agents develop opinions in two ways: For their in-neighbors they can obtain a trust observation, they then use this information to form their own opinions. For the other agents, they use the opinions of their in-neighbors they trust to update their opinions. Under the assumption that the subgraph of legitimate agents is strongly connected and each malicious agent is observed by at least one legitimate agent, we show that all legitimate agents can almost surely learn the trustworthiness of all other agents.

Our contributions can be summarized as follows: i) We present a novel learning protocol that enables the legitimate agents in the system to learn the trustworthiness of the other agents where the underlying communication network is a directed graph; 
ii) We prove that using our learning protocol, legitimate agents can learn the identities of the other agents almost surely; 
iii) We show that opinions of the agents converge in mean to the true identity of the agents; iv) We provide extensive numerical studies to show that the convergence results hold in practice for various network topologies and the number of malicious agents. 

\section{Problem Formulation}

We consider a distributed multi-agent system where agents need to collaborate in order to achieve a common task such as solving an optimization problem. 
We represent the communication graph among agents with a directed graph $\G=(\V, \E)$ where the set $\V$ represents the set of agents communicating over $\G$ with a set $\E$ of directed links. Moreover, we let $N=|V|$ be the number of agents.
If there is an edge $(i,j)\in \E$, then agent $i$ can send information to $j$, and we say that $j$ is an out-neighbor of $i$ and $i$ is an in-neighbor of $j$. We assume each agent $i$ has a self-loop $(i,i)\in E$. Moreover, for an agent $i\in V$, we define its in-neighborhood $\Nin =\{j \in V\mid (j,i)\in E\}$ and out-neighborhood $\Nout=\{j \in V \mid (i,j)\in E\}$. 
We assume that agents in the system communicate at every time step $t$.
Moreover, we assume that there might be a set $\MM \subsetneq V$, called \textit{malicious agents}, of  non-cooperative agents in the system that are either adversarial or malfunctioning. We assume that malicious agents can act arbitrarily.
We call the set of cooperative agents, that is, the set of agents outside the set $\MM$, legitimate agents, denoted by $\LL$. We have $\LL \cap \MM = \emptyset$ and $\LL \cup \MM = \V$. We say that malicious agents are untrustworthy and legitimate agents are trustworthy. We assume that the set of malicious agents $\MM$ is unknown. We wish to learn the \textit{trustworthiness} of agents in the network. We are interested in the problems where every agent receives a stochastic observation of trust from an agent that sends information during each communication round. We note that stochastic observations of trust have been developed in previous works~\cite{gil2017guaranteeing,ourTRO} and we use a similar definition here:

\begin{definition}[Stochastic Observation of Trust $\aij$]
    We denote stochastic observations of trust with $\aij(t)$ if agent $j$ sends information to agent $i$ at time $t$, and we assume that $\aij(t) \in [0,1]$. Here, $\aij(t)$ represents the stochastic value of trust of agent $j$ as observed by agent $i$. 
\end{definition}

Agents can develop opinions about trustworthiness of their in-neighbors using these stochastic trust observations over time. 
However, it is not straightforward how they can develop opinions about their out-neighbors since they have no direct observations of their trustworthiness. Next, we formalize the notion of opinion and then we discuss how to construct opinions of agents.

\begin{definition}[Opinion of Trust]
    We denote agent $i$'s opinion of trust about agent $j$ at time $t$ with  $o_{ij}(t)\in[0,1]$. We say agent $i$ trusts agent $j$ if $o_{ij}(t)\geq1/2$ and does not trust agent $j$ otherwise.
\end{definition}

We want to find a learning protocol to enable the legitimate agents to develop accurate opinions $o_{ij}(t)$ about their neighbors in directed graphs, including their out-neighbors. An example case is shown in Figure \ref{fig:multi_robot_example}.
\begin{figure}[t]
    \centering

    \subfigure[]{\includegraphics[scale=0.34]{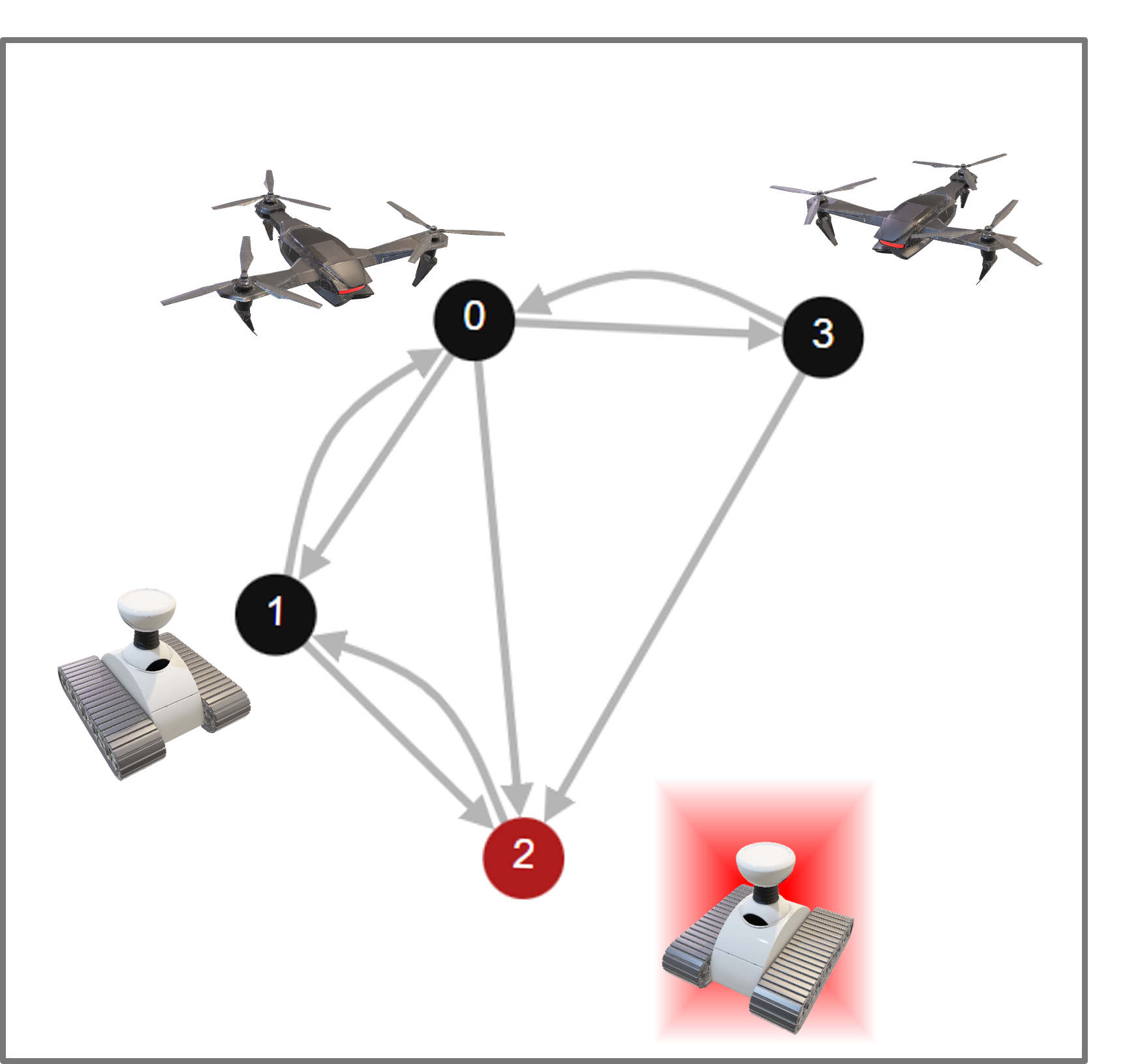}}
    \hspace{5mm}
    \subfigure[]{\includegraphics[scale=0.34]{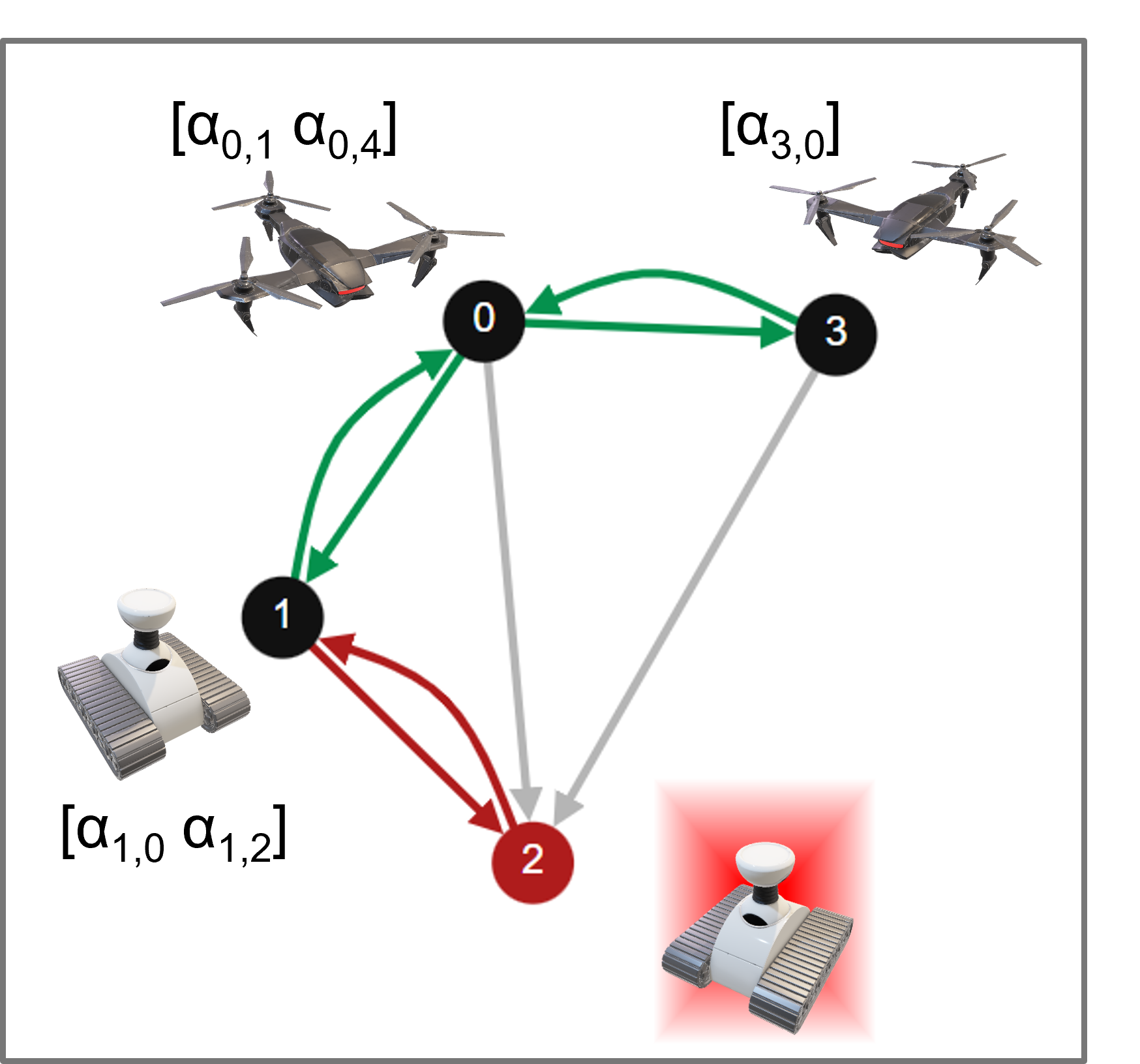}}
    \hspace{5mm}
    \subfigure[]{\includegraphics[scale=0.34]{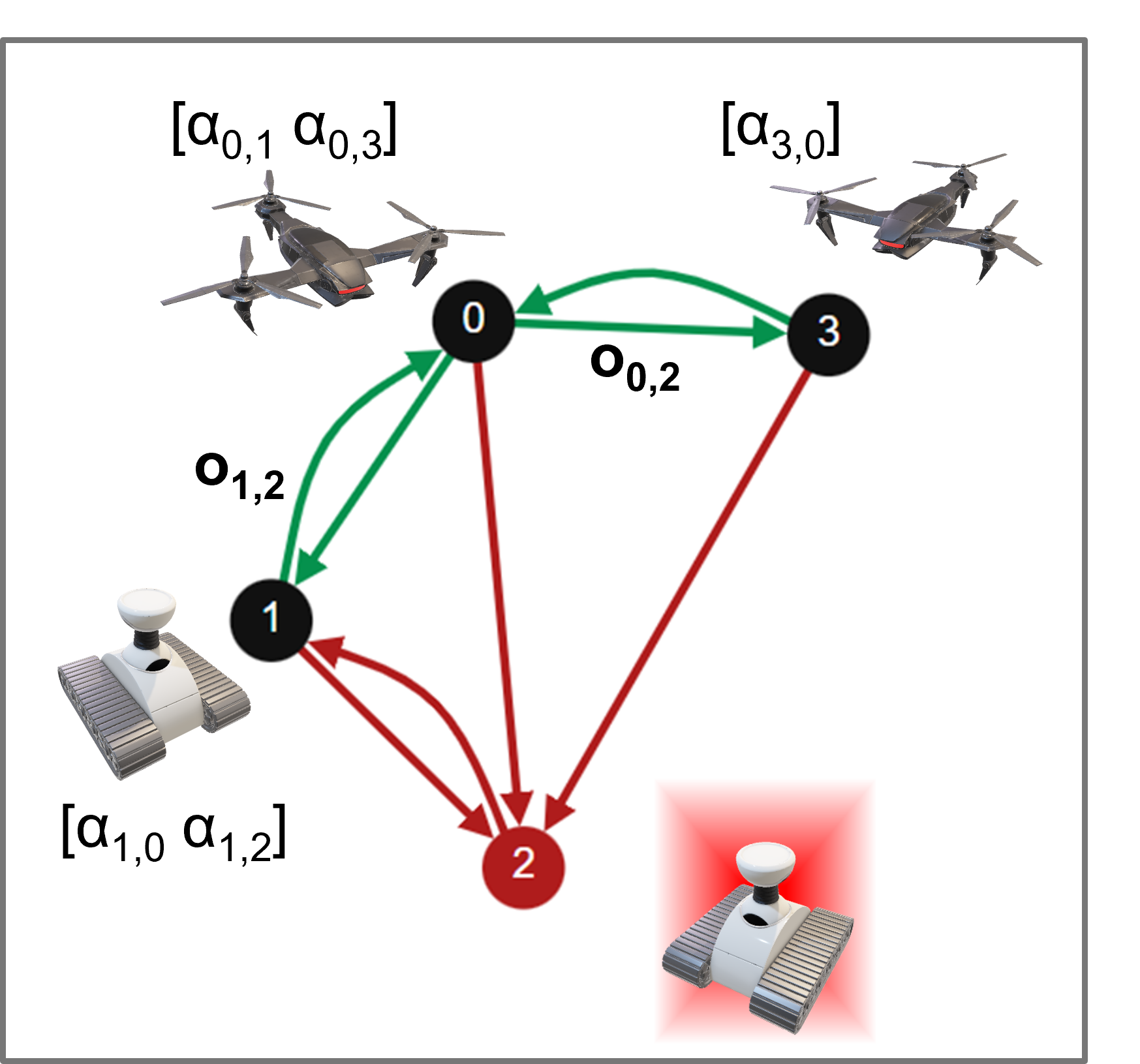}}
    
    \caption{This schematic shows our problem setup with one malicious agent shown as a red node. Various stages of learning are depicted: (a) initial state (b) agents use their direct observations to learn the trustworthiness of other agents (c) agents indirectly learn the trustworthiness of the entire network by propagating their opinions.}
    \label{fig:multi_robot_example}
    
\end{figure}
Next, we state our assumptions under which we develop our protocol.

\begin{assumption}[Connectivity of Network]\label{assumption:graph-connectivity}

\begin{enumerate}[leftmargin=*]

    \item Sufficiently connected graph: The subgraph $\GL$ induced
by the legitimate agents is strongly connected.

    \item Observation of malicious agents: For any malicious agent $j \in \MM$, there exists some legitimate agent $i \in \LL$ that observes $j$, i.e., $j \in \Nin$ for some $i \in \LL$.
\end{enumerate}
\end{assumption}

\begin{assumption}[Trust Observations]\label{assumption:trust-observations}
Suppose that the following hold:
\begin{enumerate}
    \item Homogeneity of trust variables: The expectation of the
variables $\aij(t)$ are constant for the case of malicious transmissions and legitimate transmissions, respectively, i.e., for some scalars $E_{\mM}$, $E_{\mL}$ with $E_{\mM}<0$ and $E_{\mL}>0$, $E_{\mM} = \mathbb{E}[\aij(t)]-1/2$ for all $i \in \LL, \ j \in \Nin \cap \MM$, and $E_{\mL} = \mathbb{E}[\aij(t)]-1/2$ for all $i \in \LL, \ j \in \Nin \cap \LL.$

    \item Independence of trust observations: Given the true trustworthiness of $j$, the observations $\aij(t)$ are independent for all $t$.
\end{enumerate}
\end{assumption}
Note that stochastic observations of trust satisfying Assumption~\ref{assumption:trust-observations}.1 were derived in~\cite{gil2017guaranteeing}. Additionally, we make the same Assumptions \ref{assumption:graph-connectivity}.1, \ref{assumption:trust-observations}.1, and \ref{assumption:trust-observations}.2 as in the work \cite{ourTRO}, except for the first assumption, where we require the graph to be strongly connected instead of connected since we deal with directed graphs. Assumption \ref{assumption:graph-connectivity}.2 is new and necessary since it is not possible to learn the legitimacy of an agent if no other agent is observing that agent. This requirement shows up in the analysis later on. We formalize the problem that we are aiming to solve in this paper as follows:

\begin{problem}
Let $i$ be a legitimate agent and let $q$ be an arbitrary agent in the system. Assume that stochastic observations of trust are available and Assumption \ref{assumption:graph-connectivity} and Assumption \ref{assumption:trust-observations} hold. We want to find a learning protocol such that for all legitimate agent $i\in \LL$ and for all agents $q\in \V$, $o_{iq}(t)$ converges to $1$ if $q\in \LL$ and $0$ if $q\in \MM$ almost surely.
\end{problem}

\section{Learning Protocol}

In this section we introduce our learning protocol. Let each agent $i$ store a vector of trust $o_i(t)$ at time t, where $o_i(t)$ is an $N\times1$ column vector.
Let $o_{ij}(t)$ denote the $j$th component of $o_i(t)$. The value $o_{ij}(t)$ represents agent $i$'s opinion about the node $j$ where a higher $o_{ij}(t)$ indicates that agent $i$ trusts agent $j$ more. Let $\bij(t)$ represent an aggregate trust value for the link $(j,i)$ at time $t$. Following \cite{ourTRO}, we define $\bij(t)$ as 

\begin{equation}
    \bij(t)=\sum_{k=0}^{t}(\aij(k)-1/2)
\end{equation}

for all $j\in \Nin$ and we define $\beta_{ii}(t)=1$ for all $t$.
Using the aggregated stochastic trust value $\bij(t)$, a legitimate agent $i$ decides on its trusted in-neighbor set by defining $
\Nin(t)=\{j\in \Nin \mid \bij(t)\geq 0\}
.$
In our learning protocol, an agent $i$ shares $o_i(t)$ with its out-neighbors. A legitimate agent $i$ determines its vector of $o_i(t)$ after receiving $o_j(t-1)$ from all of its in-neighbors $j\in \Nin$ using the following update rule:
\begin{equation}
\label{eq:trust_vector_update_rule}
    o_{iq}(t)= 
    \begin{cases}
      1        &\text{if} \ q\in \Nin\ and \ \beta_{iq}(t)\geq 0\ \\
      0        & \text{if}\ q \in \Nin \ and \ \beta_{iq}(t)<0.  \\
      \sum_{j \in \Nin(t)} \frac{o_{jq}(t-1)}{|\Nin(t)|} & \text{if}\ q \notin \Nin \ 
    \end{cases}
\end{equation}
Every legitimate agent $i$ initializes its opinion vector with vector $o_i(0)$ with all ones, meaning that in the beginning, they trust everyone in the network. However, this choice of initialization is arbitrary and as it does not affect our results. A legitimate agent $i$ decides on its trusted out-neighbor set by defining $\Nout(t)=\{j\in \Nout \mid o_{ij}(t)\geq 1/2\}$.

Notice that the trust vector $o_i(t)$ is in $[0,1]^{N}$ by definition. We assume that malicious agents can decide its trust vector $o_i(t)$ arbitrarily. With this protocol, legitimate agents use only the stochastic observations of trust $\alpha_{ij}$ to determine the legitimacy of their in-neighbors. For the other nodes, they use the opinions of their trusted in-neighbors to form their opinion. 

\section{Analysis}
Recall that agents either directly observe an agent and develop their own opinions using their observations, or they use the opinions of others to generate an opinion about an agent. In our analysis, we first show that all legitimate agents learn their in-neighbors such that their trusted in-neighbors are the same as their legitimate in-neighbors. Learning in-neighbors allow agents to propagate this information to others and also stop the inflow of information from any malicious agent. Then, we analyze the propagation of information after legitimate agents learned their in-neighbors. To do this, we write the update rule of trustworthiness about an agent in matrix form, and show that the effect of the error introduced by malicious agents is asymptotically eliminated. More precisely, we show that estimated trust values converge in mean and almost surely to true trust values (1 for legitimate, 0 for malicious agents).
\subsection{Notation}
Let $|S|$ denote the cardinality of set $S$. Let $[W]_{ij}$ denote entry  in row $i$ and column $j$ of matrix $W$. For some agent $j$ and set $S$, define the indicator function $\mathbf{1}_{\{j \in \mathcal{S}\}}$ as:
$$\mathbf{1}_{\{j \in \mathcal{S}\}} = \begin{cases}
 1 &\text{if $j \in \mathcal{S}$} \\
  0 &\text{otherwise}
\end{cases}.$$ 
We also use the same notation for indicator vectors when the size of the vector is clear from the context.
\subsection{Learning Trustworthiness}
Since agents use their trusted in neighbors in their updates, we start by showing that agents learn the legitimacy of their in-neighbors. This will be useful later to show that the protocol converges to the desired state.
\begin{lemma}
    \label{lemma:learning_in_neighbors}
    There exists a random finite time $T_f$ such that the following holds almost surely
    \begin{flalign}
        \bij(t) \geq 0 \ \text{for all} \ t\geq T_f \ \text{and} \ i\in \LL, j\in \Ni^{\rm in}\cap \LL \\ \nonumber
        \bij(t) < 0 \ \text{for all} \ t\geq T_f \ \text{and} \ i\in \LL, j\in \Ni^{\rm in}\cap \MM
    \end{flalign}
    
\end{lemma}
\begin{proof}
    Follows directly from \cite[Proposition 1]{ourTRO}
\end{proof}

\begin{corollary}
    \label{cor:learning_in_neighbors}
    There exists a random finite time $T_f$ such that for all $t\geq T_f$ and for all legitimate agents $i$, trusted in-neighbor set consist of all legitimate neighbors of the agent $i$, that is: 
    $$\Ni^{\rm in}(t) = \Ni^{\rm in}\cap\LL.$$
\end{corollary}
\begin{proof}
    Follows directly from Lemma \ref{lemma:learning_in_neighbors} and the update rule of the learning protocol given by \eqref{eq:trust_vector_update_rule}.
\end{proof}

Notice that corollary \ref{cor:learning_in_neighbors} shows that every legitimate agent can learn its in-neighbors correctly. Now, let $q\in V$ be a an arbitrary but fixed agent in the network. Our goal is to show that all legitimate agents learn the identity of $q$. This process requires information to propagate from agents receiving trust information directly from $q$ to other agents in the network, which motivates the following definitions:
 
 Define $\mD_q \subseteq \LL$ to be the subset of legitimate agents directly observing $q$, i.e. $\mD_q\triangleq\No_q\cap \LL$. Similarly, define $\mathcal{C}_q \subseteq \LL$ be the subset of legitimate agents not observing $q$, i.e. $\mC_q\triangleq\LL \backslash D_q$. These sets are illustrated in Fig. \ref{fig:graph_c_d}.
\begin{figure}
    \centering
    \subfigure[Example network]{\includegraphics[width=0.32\textwidth]{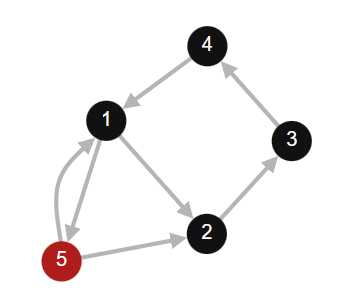}}
    \subfigure[$q=2$]{\includegraphics[width=0.32\textwidth]{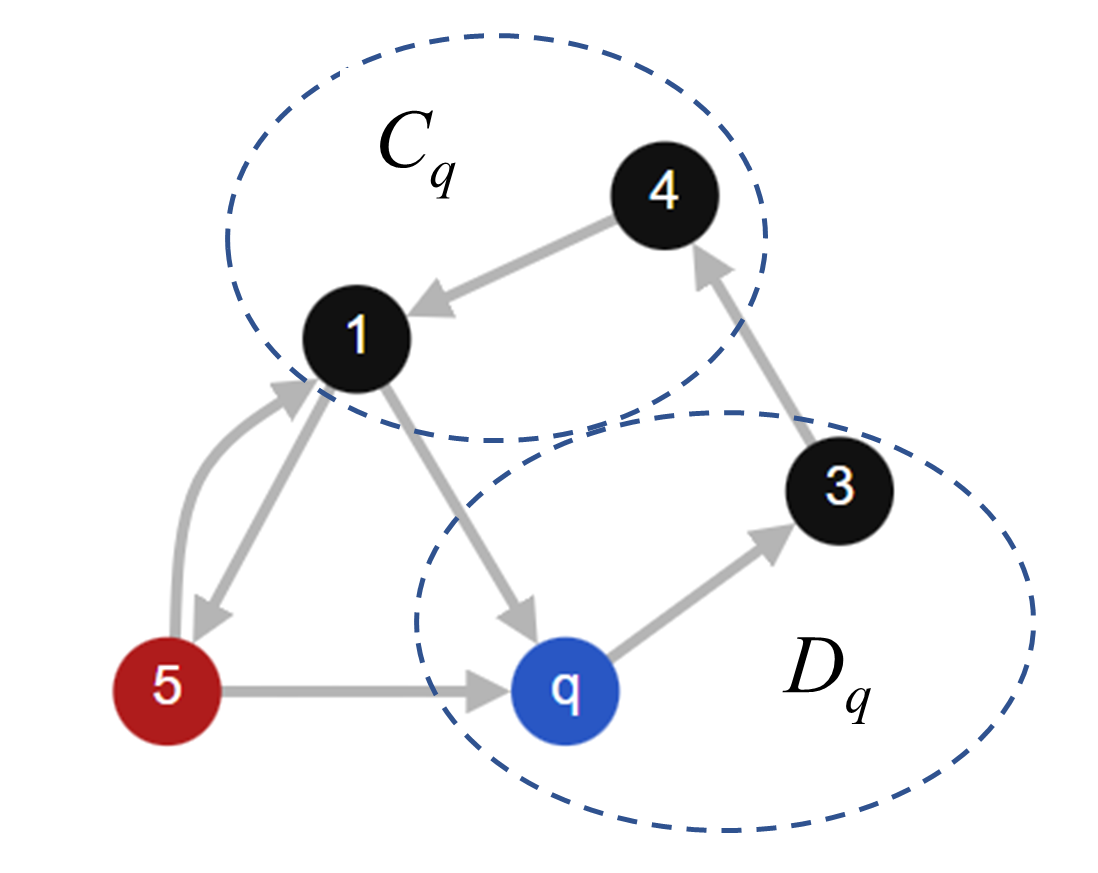}}
    \subfigure[$q=5$]{\includegraphics[width=0.32\textwidth]{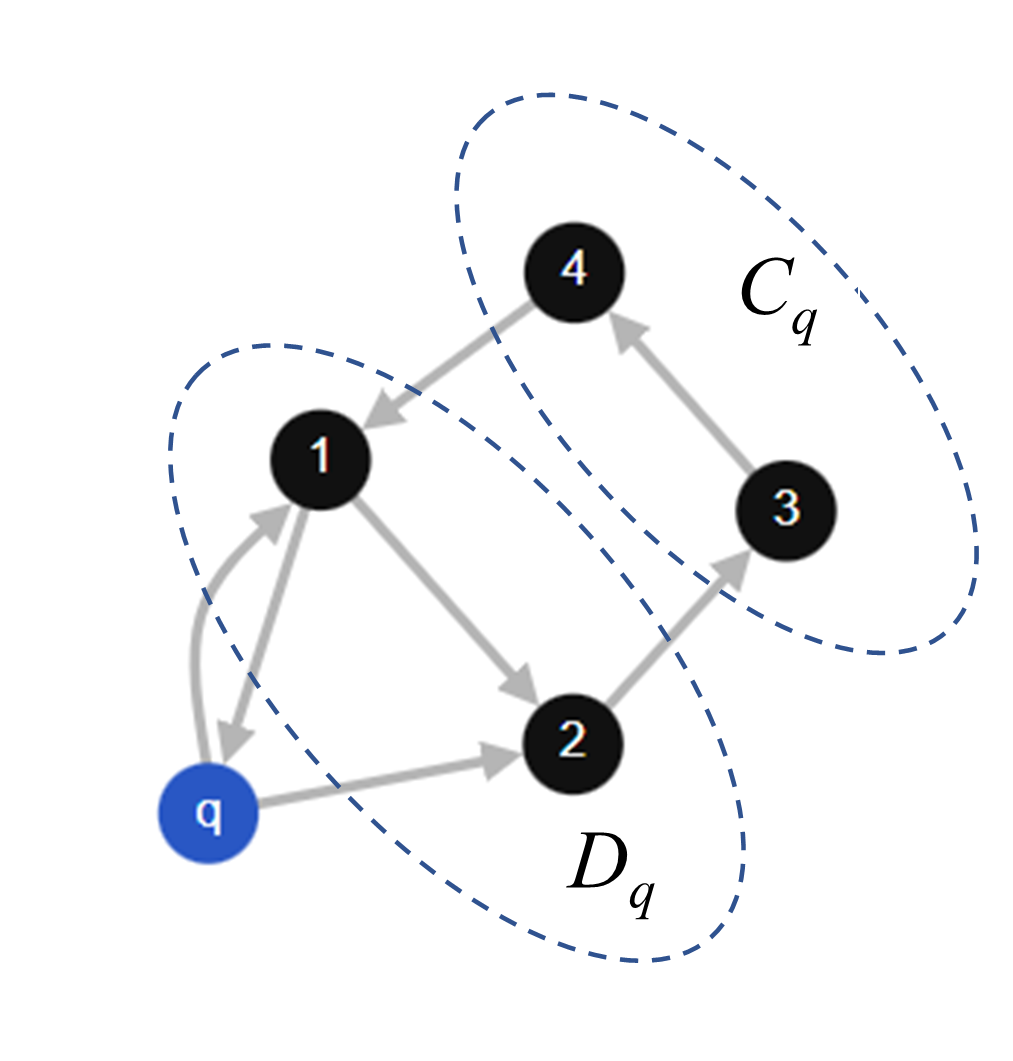}}
    \caption{
    (a) Network with four legitimate and one malicious nodes. Legitimate nodes are black and the malicious node is red.
    (b) Learning dynamics for agent $q=2$: Both agents $2$ and $3$ directly observe agent $2$, so $2, 3\in \mD_q$. $1$ and $4$ are in the set $\mC_q$ since they do not directly observe agent $2$. (c) Learning dynamics for agent $q=5$: Both agents $1$ and $2$ directly observe agent $5$, so $1,2\in \mD_q$. $3$ and $4$ are in $\mC_q$ since they do not observe $5$.}
    \label{fig:graph_c_d}
\end{figure}

These sets are defined for the sake of analysis and are not assumed to be known in practice. Notice that the set $\mD_q$ of observing agents is non-empty. This follows directly from Assumption \ref{assumption:graph-connectivity}.2, since there exists at least one legitimate agent $i \in \No_q$. 
On the other hand, the set $\mC_q$ can be empty if all agents are directly observing $q$. In that case, all legitimate agents will eventually learn the identity of $q$ by Corollary \ref{cor:learning_in_neighbors}. 

Now, we analyze the evolution of $o_{iq}(t)$ by writing the evolution of opinions about agent $q$ in matrix form. Let $u_q=|\mC_q|$, i.e. the number of agents not observing $q$. Without loss of generality, reorder the indices of agents such that $\mC_q=\{1,2,\dots, u_q\}$, and $\mD_q=\{u_q+1,\dots, |\LL| \}$.
 We denote the vector of trust estimates of legitimate agents about the agent $q$ by $o_{\mC_q}(t)$
$= \begin{bmatrix}o_{1,q}(t) & \dots & o_{u_q,q}(t)\end{bmatrix}^T$ for the agents in the set $\mC_q$ and $o_{\mD_q}(t)=\begin{bmatrix}o_{u_q+1,q}(t) &\dots&o_{|\LL|,q}(t) \end{bmatrix}^T$ for the agents in the set $\mD_q$. Finally, we denote the vector of trust estimates of malicious agents about the agent $q$ by $o_{\MM_q}(t)= \begin{bmatrix}o_{|\LL|+1,q}(t) & \dots & o_{N,q}(t)\end{bmatrix}^T$. Take an arbitrary agent $i\in \mC_q$. Using these reordered indices, we can rewrite the learning protocol as:

\begin{flalign}
\label{eq:matrix_learning_protocol}
    o_{iq}(t) &= 
      \sum_{j \in \Nin(t)} \frac{o_{jq}(t-1)}{|\Nin(t)|}  \\
      &= \sum_{j\in \mC_q} [W_q(t)]_{ij} o_{jq}(t-1) + \sum_{j \in \mD_q}[W_q(t)]_{ij} o_{jq}(t-1) + \sum_{j\in \MM}[W_q(t)]_{ij} o_{jq}(t-1),
\end{flalign}
where $[W_q(t)]_{ij}=\frac{1}{|\Nin(t)|}$ if $j\in \Nin(t)$ and $[W_q(t)]_{ij}=0$ otherwise. Here, $W_q(t)$ is a row-stochastic matrix with size $u_q \times N$. Then, we can divide $W_q(t)$ into three parts based on the sets $\mC_q$, $\mD_q$, $\MM$ as $W_q(t) = [W_{\mC_q}(t)\; W_{\mD_q}(t)\; W_{\MM_q}(t)]$ where the matrices $W_{\mC_q}(t)$, $W_{\mD_q}(t)$, $W_{\MM_q}(t)$ have sizes $u_q \times u_q$, $u_q\times |\mD_q|$, and $u_q \times |\MM|$ respectively. With this representation, we can express the update rule \eqref{eq:matrix_learning_protocol} in the matrix form as:
\begin{flalign}
    \label{eq:learning_matrix_form}
    o_{\mC_q}(t)=\begin{bmatrix}W_{\mC_q}(t) & W_{\mD_q}(t) & W_{\MM_q}(t)\end{bmatrix}\begin{bmatrix}
o_{\mC_q}(t-1) \\
o_{\mD_q}(t-1)\\
o_{\MM_q}(t-1)
\end{bmatrix},
\end{flalign}

Recall that there exists some random finite time $T_f$ such that all legitimate agents learn their in-neighbors correctly. Until the system reaches time $T_f$, malicious agents can affect the learning dynamics. Nevertheless, we will show that the legitimate agents can recover from that effect after reaching time $T_f$. Now, we focus our analysis on the system dynamics after time $T_f$.

\begin{lemma}
    \label{lemma:t_geq_tf}
    For $t \geq T_f$, the following hold almost surely:
    \begin{enumerate}
        \item[L\ref{lemma:t_geq_tf}.1] The matrix representing the contribution of malicious agents $W_{\MM_q}(t)=0$
        \item[L\ref{lemma:t_geq_tf}.2] $W_{\mC_q}(t)=\overline{W}_{\mC_q}$ for some constant matrix $\overline{W}_{\mC_q}$
        \item[L\ref{lemma:t_geq_tf}.3] $W_{\mD_q}(t)=\overline{W}_{\mD_q}$ for some constant matrix $\overline{W}_{\mD_q}$
        \item[L\ref{lemma:t_geq_tf}.4] $o_{\mD_q}(t)=\ind$.
    \end{enumerate}
\end{lemma}
\begin{proof}
Assuming that $t\geq T_f$, by Corollary \ref{cor:learning_in_neighbors}, we have that $\Nin(t) \cap \MM=\emptyset$ for $i\in \LL$. Therefore, if $i\in \mC_q$ and $j$ is a malicious agent, then $j\not \in \Nin(t)$. By the definition of $W_q$, we have $[W_q(t)]_{ij}=0$ for all malicious $j$ as desired. Similarly, $\Nin(t)=\Nin\cap \LL$, so the update matrices $W_{\mC_q}(t)$ and $W_{\mD_q}(t)$ are constant for $t \geq T_f$. Finally, L\ref{lemma:t_geq_tf}.4 follows directly from Lemma~\ref{lemma:learning_in_neighbors}.
\end{proof}
\begin{remark} \label{remark:row-stochastic-update-matrix}
The matrix $[\overline{W}_{\mC_q}\; \overline{W}_{\mD_q}]$ is row stochastic.
\end{remark}
This follows from the fact that $W_q(t)$ is a row-stochastic matrix and that $W_{\MM_q}(t)$ is zero. Since agents in $\mD_q$ have already learned the trust of agent $q$ after time $T_f$, we now focus on the agents in $\mC_q$. For all $t\geq T_f+1$, we can describe the evolution of $o_{\mC_q}(t)$ as follows:
\begin{flalign}
\label{eq:pc_after_tf}
   o_{\mC_q}(t) = \overline{W}_{\mC_q}o_{\mC_q}(t-1)+\overline{W}_{\mD_q}o_{D_q}(t-1)
\end{flalign}
 
We want $o_{\mC_q}(t)=\ind$, i.e., $o_{\mC_q}(t)$ should be equal to a vector of ones if $q \in \LL$ and a vector of zeros if $q \in \MM$. We can define the error in the estimation of legitimate agents in $\mC_q(t)$ about the identity of the agent $q$ at time $t$ as:
\begin{flalign}
    \label{eq:delta_c_def}
    \errc(t) = o_{\mC_q}(t)-\ind
\end{flalign}
We want to show that $\norm{\errc(t)}\to 0$ as $t$ goes to infinity. Using \eqref{eq:pc_after_tf} we can represent $\errc(t)$ as
\begin{flalign}
    \label{eq:delta_recursive}
    \errc(t) &= \overline{W}_{\mC_q}o_{\mC_q}(t-1)+\overline{W}_{\mD_q}o_{\mD_q}(t-1) - \ind \nonumber \\
    &\overset{(a)}{=} \overline{W}_{\mC_q}o_{\mC_q}(t-1)+\overline{W}_{\mD_q}o_{\mD_q}(t-1) - (\overline{W}_{\mC_q}\ind + \overline{W}_{\mD_q}\ind) \nonumber \\
    &= \overline{W}_{\mC_q}(o_{\mC_q}(t-1)-\ind)+\overline{W}_{\mD_q}(o_{\mD_q}(t-1)-\ind) \nonumber \\
    &\overset{(b)}{=}\overline{W}_{\mC_q}(o_{\mC_q}(t-1)-\ind) \nonumber \\
    &= \overline{W}_{\mC_q}\errc(t-1),
\end{flalign}
where $(a)$ follows from the fact that $[\overline{W}_{\mC_q}\; \overline{W}_{\mD_q}]$ is row stochastic and $(b)$ follows from $o_{\mD_q}(t-1)=\ind$. By using \eqref{eq:delta_recursive} recursively, we obtain 
\begin{flalign}
\label{eq:delta_power}
    \errc(t) = \overline{W}_{\mC_q}^{t-T_f}\errc(T_f)
\end{flalign}
Now, we can bound the error norm:
\begin{flalign}
    \label{eq:delta_norm}
    \norm{\errc(t)}\leq \norm{\overline{W}_{\mC_q}^{t-T_f}}\norm{\errc(T_f)}.
\end{flalign}
Here, $\|\Delta_{\mC_q}(T_f)\|$ includes the error introduced by malicious agents before all agents learn their in-neighbors. Since the convergence of the error term $\|\Delta_{\mC_q}(t)\|$ depends on the convergence of $\overline{W}_{\mC_q}$, we analyze the matrix $\overline{W}_{\mC_q}$ next.

\subsection{Convergence of Weakly Chained Substochastic Matrices}
Now, we aim to show that $\overline{W}_{\mC_q}$ is \textit{convergent}, i.e. $\|\overline{W}_{\mC_q}^t\| \to 0$ as $t\to \infty$. In this part, we will show that $\overline{W}_{\mC_q}$ belongs to a family of \textit{convergent} substochastic matrices called \textit{weakly chained substochastic} matrices. This will conclude that the error term goes to $0$. First, we give some definitions.
\begin{definition}
Digraph of matrix: Let the square matrix $W\in \mathbb{R}^{n\times n}$ be non-negative, i.e. $W_{ij}\geq 0$ for all $i,j$. Hence, the graph of  $W$, denoted by $\GW{W}=(\VW{W},\EW{W})$ is the graph such that $\VW{W}=\{1,\dots,n\}$ and for all $i,j \in \{1,\dots,n\}$, $(i,j) \in \EW{W}$ if and only if $W_{ij}>0$.
\end{definition}
To analyze the convergence properties of $\overline{W}_{\mC_q}$, we define the index of contraction following \cite{azimzadeh2019fast}
\begin{definition}
Index of contraction: Let the matrix $W\in \mathbb{R}^{n\times n}$ be substochastic. Define the set $\hat{J}(W) \triangleq \{1\leq i \leq n : \sum_{j=1}^{n} W_{ij} < 1 \},$ and let the set $\hat{K}_i(W)$ be the set of all paths\footnote{We use path instead of walk in contrast to \cite{azimzadeh2019fast} in our definition, however these definitions are equivalent.} in the digraph of $W$ from $i$ to all $j\in \hat{J}(W)$. The index of contraction $\widehat{con}W$ associated with matrix $W$ is defined as:
\begin{flalign}
    \label{eq:index_of_contraction}
    \widehat{con}W \triangleq \max\left\{0, \sup_{i\not\in \hat{J}(W)}\left\{ \inf_{\omega \in \hat{K}_i(W)}\left\{|\omega| \right\}\right\}\right\},
\end{flalign}
where $|\omega|$ denotes the length of the path $\omega$. Also, we follow the conventions that $\inf \emptyset = \infty$ and $\sup \emptyset = -\infty$. Here, if all rows of $W$ sum to less than one, we have $|\hat J (W)|=n$. This implies that the supremum over $i\not \in \hat J(W)$ is $-\infty$, therefore, $\widehat{con}W=0$. Similarly, $\widehat{con}W$ is infinite if $\hat K_i(W)$ is empty, meaning there is no path from some row $i\not \in \hat J(W)$ to any row that sums to less than one.
\end{definition} \cite[Corollary 2.6]{azimzadeh2019fast} shows that a square substochastic matrix $W$ is convergent if and only if $\widehat{con}W$ is finite. We show example matrices with different contraction indices in Figure \ref{fig:contraction_index}.
\begin{figure}
    \centering
    \subfigure[$\widehat{con}W=3$]
    {\includegraphics[width=0.20\textwidth]{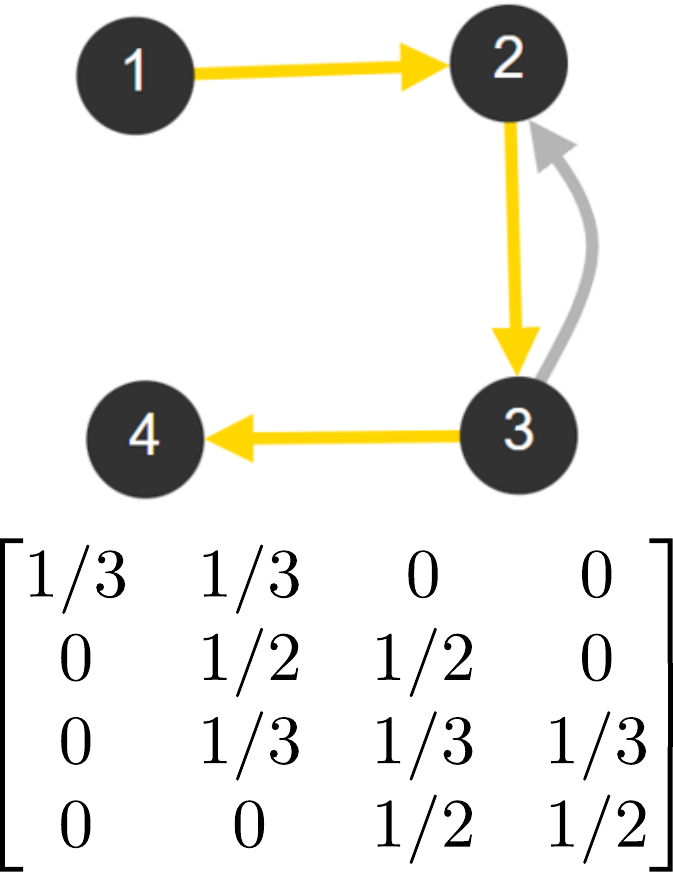}}
    \hspace{15mm}
    \subfigure[$\widehat{con}W=2$]
    {\includegraphics[width=0.20\textwidth]{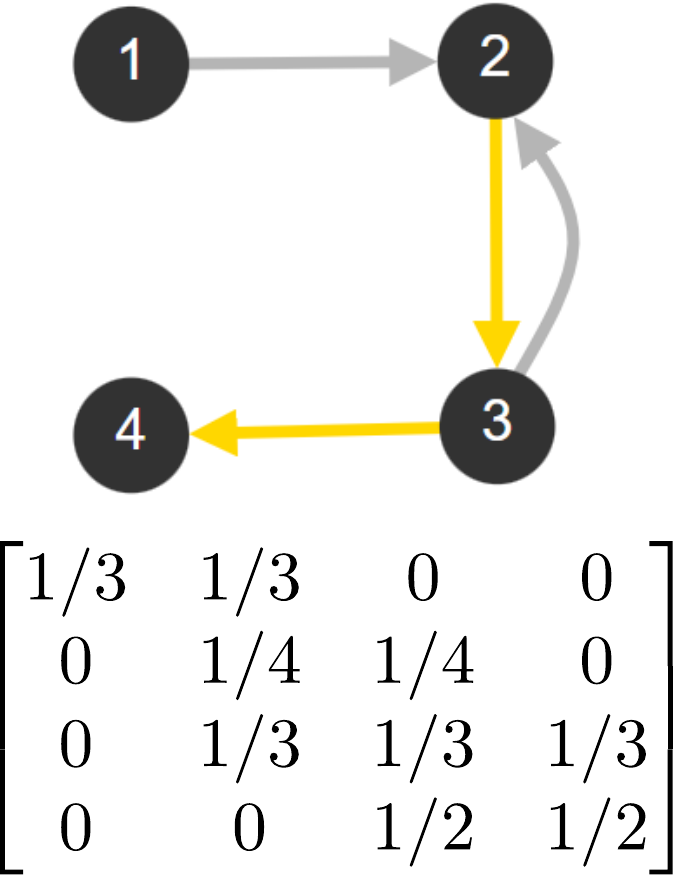}}
    \hspace{15mm}
    \subfigure[$\widehat{con}W=\infty$]
    {\includegraphics[width=0.20\textwidth]{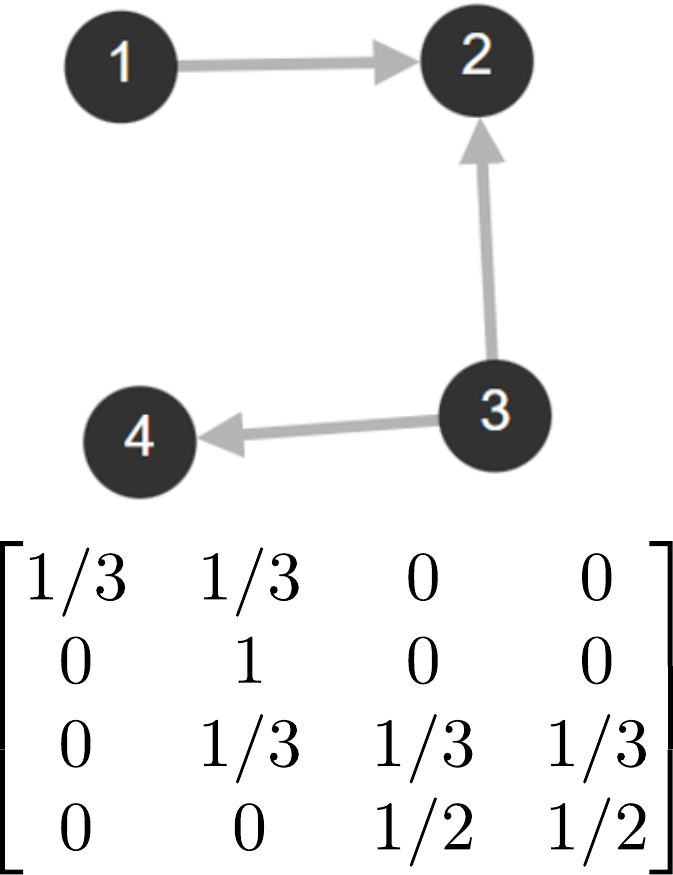}}
    \caption{Three matrices with different contraction indices and the corresponding graphs. The path achieving the contraction index is given in yellow. (a) The only row that sums to less than one is row $1$. Therefore, the path with the maximum length from agent $1$ to another agent has a length of $3$. (b) Row $2$ also sums to less than one. Therefore, the longest path is the one from agent $2$ to $4$. (c) The only row that sums to less than one is row $1$. Since there is no path from agent $1$ to agent $3$ or $4$, the index of contraction is $\infty$.}
    \label{fig:contraction_index}
\end{figure}
We call a substochastic matrix with finite contraction index \textit{weakly chained substochastic matrix}.  
\begin{remark}\label{remark:weakly-chained-substochastic}
    Matrix $W$ is a weakly chained substochastic matrix if and only if for all rows $i$ that are not in the set $\hat{J}(W)$, set $\hat{K}_i(W)$ is non-empty, i.e  there is a path $i \rightarrow i_1 \rightarrow \dots \rightarrow i_j$ in $\GW{W}$ such that row $i_j$ sums to less than one. Moreover, a weakly chained substochastic matrix is convergent.
\end{remark}
This remark follows directly from the definition of the index of contraction and \cite[Corollary 2.6]{azimzadeh2019fast}.

The following sequence of results will show that $\overline{W}_{\mC_q}$ is weakly chained substochastic. We will first establish a relation between the graph that describes the network and the digraph of $\overline{W}_{\mC_q}$. In particular, we will establish that the links $\EW{\overline{W}_{\mC_q}}$ in the digraph of $\overline{W}_{\mC_q}$ are the inversion of links in the original graph. Then use assumptions of strong connectivity and existence of a directly observing agent to conclude that $\overline{W}_{\mC_q}$ is weakly chained substochastic.

\begin{lemma} \label{lemma:weight-matrix-path}
Let $\overline{W}_{\mC_q} \in \mathbb{R}^{u_q\times u_q}$ be defined as before, and let $\G_{\mC_q}$ be the subgraph of $\G_{\mathcal{L}}$ induced by the set of agents $\mC_q$. Then, $(i,j)\in \GW{\overline{W}_{\mC_q}}$ if and only if $(j,i)\in \G_{\mC_q}$. In other words, $\G_{\mC_q}$ is the digraph of $\overline{W}_{\mC_q}^T$.
\end{lemma}

\begin{proof} 
Let $(i,j) \in \GW{\overline{W}_{\mC_q}}$. Then, by definition of digraph of a matrix, we have that $[\overline{W}_{\mC_q}]_{ij}>0$. So, it means that agent $i \in \mC_q$ is receiving information from agent $j \in \mC_q$, by the learning protocol \eqref{eq:matrix_learning_protocol}. Thus, there must be an edge $(j,i)$ in $\G_{\mC_q}$. Similarly, if $(j,i)$ is an edge in $\G_{\mC_q}$, then agent $j$ is an in-neighbor of agent $i$. So, agent $i$ is receiving informattion from agent $j$, which means that $[\overline{W}_{\mC_q}]_{ij}>0$. 
\end{proof}

\begin{corollary} \label{cor:path-inversion-in-matrix}
If there is a path $v_1 \to v_2\to \dots \to v_l$ in $\G_{\mC_q}$, then there is a path $v_l \to v_{l-1}\to \dots \to v_1$ in $\GW{\overline{W}_{\mC_q}}$.
\end{corollary}

\begin{proof}
If there is an edge $v_i \to v_{i+1}$ in $\G_{\mC_q}$, then by Lemma \ref{lemma:weight-matrix-path}, there exists and edge $v_{i+1}\to v_{i}$ in $\GW{\overline{W}_{\mC_q}}$. Since this holds for each $i=1,\dots, l-1$, $v_l\to v_{l-1}\to \dots v_{1}$ is a path in $\GW{\overline{W}_{\mC_q}}$.
\end{proof}

\begin{theorem}
    \label{thm:W_C_weakly_chained}
    For all agents $q$, given that the set $\mC_q$ is non-empty, the update matrix $\overline{W}_{\mC_q}$ is a weakly chained substochastic matrix. Moreover, $\overline{W}_{\mC_q}$ is convergent.
\end{theorem}

\begin{proof}
 Let $i \in \mC_q$. If agent $i$ has a neighbor $d \in \mD_q$ directly observing agent $q$, row $i$ must sum up to less than one since agent $i$ receives information from $d$ and $d\not \in \mC_q$. So, $i \in \hat J(\overline{W}_{\mC_q})$.
 
 Now, assume agent $i$ doesn't have a directly observing neighbor, i.e. $i\not \in \hat J(\overline{W}_{\mC_q})$. We know that there exists some agent $d\in \mD_q$ that directly observes agent $q$ by Assumption~\ref{assumption:graph-connectivity}.1 and Assumption~\ref{assumption:graph-connectivity}.2. By Assumption~\ref{assumption:graph-connectivity}.1, the subgraph induced by legitimate agents are strongly connected, so there exists a path 

$$d=i_0 \to i_1 \to i_2 \to \dots i_l\to i$$

\noindent in $\G_{\mathcal{L}}$ for each $i_j \in \mathcal{L}$ where each arrow denotes a directed edge.\footnote{$l\geq 1$ since agent $i$ does not have a directly observing neighbor} Now, choose the largest $j$ such that $i_j\in \mD_q$, and consider the path
$$i_{j}\to i_{j+1}\to \dots \to i_{l}\to i$$
Here, since $j$ is chosen as the largest $j$ s.t. $i_j \in \mD_q$, we have that $i_{j+1}, \dots, i_{l}, i \in \mC_q$. Moreover, we assumed $i\not \in \hat J(\overline{W}_{\mC_q})$, so $j < l$ since $i$ does not have a directly observing neighbor.

Now, we know, $i_{j+1}$ has a neighbor directly observing $q$, i.e. $i_j$. Therefore, row $i_{j+1}$ of $\overline{W}_{\mC_q}$ sums to less than $1$, meaning that $i_{j+1} \in \hat J(\overline{W}_{\mC_q})$. From Corollary~\ref{cor:path-inversion-in-matrix}, there exists a path 
$$i \to i_l \to i_{l-1}\to \dots \to i_{j+2}\to i_{j+1}$$
in the graph $\GW{\overline{W}_{\mC_q}}$. So, in the digraph $\GW{\overline{W}_{\mC_q}}$ there exists a path from $i$ to a row summing to less than one, $i_{j+1}$, as desired. Hence, $\hat K_i(\overline{W}_{\mC_q})$ is non-empty for $i\not \in \hat J(\overline{W}_{\mC_q})$. Therefore, $\overline{W}_{\mC_q}$ is weakly chained substochastic and convergent by Remark~\ref{remark:weakly-chained-substochastic}.
\end{proof}
\begin{corollary}
\label{cor:c_q_almost_surely}
    For all agents $q\in \V$ where the set $\mC_q$ is non-empty, $o_{\mC_q}(t)$ almost surely converges to $\ind$ where $\ind$ is a vector with all values equal to $1$ if $q \in \LL$ and to $0$ if $q \in \MM$.
\end{corollary}
\begin{proof}
Remember that the error is defined as $\errc(t) = o_{\mC_q}(t)-\ind$. By Corollary \ref{cor:learning_in_neighbors} we know that there exists a finite time $T_f$ such that for all $t\geq T_f+1$ we have
\begin{equation}
    \norm{\errc(t)}\leq \norm{\overline{W}_{\mC_q}^{t-T_f}}\norm{\errc(T_f)} \tag{\eqref{eq:delta_norm}}.
\end{equation}
Since both $o_{\mC_q}(t)$ and $\ind$ are in $[0,1]^{u_q}$, we have $\norm{\errc(T_f)}\leq 
\sqrt{u_q}$. By Theorem \ref{thm:W_C_weakly_chained}, we have that $\norm{\overline{W}_{\mC_q}^{t-T_f}} \to 0$. Therefore, $\norm{\errc(t)}\to 0$ almost surely.
\end{proof}

\subsection{Main Results}
In this part, we present our main results which show that all legitimate agents can learn the trustworthiness of of all agents in the system. Let $\ol(t)$ denote the trustworthiness estimation of all legitimate agents about an agent $q$ at time $t$. We show that this vector converges to $\ind$.
\begin{theorem}[Convergence to the true trust vector almost surely]
    \label{thm:as_conv}
    For all agents $q\in \V$, $\ol(t)$ converges almost surely to the true trust vector $\ind$, where $\ind$ is an $|L|\times 1$ vector with all of its values equal to $1$ if $q \in \LL$ and equal to $0$ if $q \in \MM$.
\end{theorem}
\begin{proof}
    Without loss of generality, we reorder the indices of agents such that $\mC_q=\{1,2,\dots, u_q\}$, and $\mD_q=\{u_q+1,\dots, |\LL| \}$ where $\mC_q$ is the set of legitimate agents not observing $q$ and $\mD_q$ is the set of legitimate agents directly observing $q$. We have two different cases where the set $\mC_q$ is empty and non-empty. First assume that $\mC_q$ is empty. We know that $\ol(t) = o_{\mD_q}$. There exists a finite time $T_f$ such that for all $t\geq T_f+1$ we have $\ol(t)=\ind$ by Lemma \ref{lemma:t_geq_tf}. Hence, $\ol(t)$ converges to $\ind$ almost surely.
    
    Now, assume that $\mC_q$ is non-empty. Hence, we can represent $\ol(t)$ as $\ol(t) = \begin{bmatrix}
    o_{\mC_q}(t)\\
    o_{\mD_q}(t)
    \end{bmatrix}$. Define $\Delta_{\LL_q}(t)=\ol(t)-\ind$. Using the triangle inequality we obtain 
    \begin{flalign*}
        \norm{\Delta_{\LL_q}(t)} &\leq \norm{o_{\mC_q}(t)-\ind} + \norm{o_{\mD_q}(t)-\ind} \\
        &=\norm{\errc(t)} + \norm{o_{\mD_q}(t)-\ind},
    \end{flalign*}
    where $\errc(t)$ is the same one with \eqref{eq:delta_c_def}. Now, assume that $t\geq T_f+1$. Then we have $\norm{o_{\mD_q}-\ind}=0$ by Lemma \ref{lemma:t_geq_tf}. Moreover, by Corollary \ref{cor:c_q_almost_surely}, we have that $\norm{\errc(t)}\to 0$ almost surely. Hence, we can conclude that $\norm{\Delta_{\LL_q}(t)}\to 0$ and $\ol(t)$ converges to $\ind$ almost surely.
\end{proof}

\begin{theorem}[Convergence in mean to the true trust vector]
\label{thm:convMean}
    For all agents $q\in \V$  and $r\geq 1$, $\ol(t)$ converges in mean to the true trust vector $\ind$. That is,
    \begin{equation}
        \lim_{t \to \infty} \E[\norm{\ol(t)-\ind}^r] = 0.
    \end{equation}
\end{theorem}
\begin{proof}
     Since $\ol(t) \in [0,1]^{|\LL|}$, $\norm{\ol(t)}_2\leq \sqrt{|\LL|}$. Also we have $\ind$ in $[0,1]^{\LL}$. Then, using the triangle inequality we get $\norm{\ol(t)-\ind}^r\leq (\sqrt{|\LL|})^r < \infty$.
     We can apply the dominated convergence theorem \cite{cinlarErhan2011PaS} to conclude our proof since $\ol(t)$ converges to $\ind$ almost surely by Theorem \ref{thm:as_conv}.
\end{proof}
Finally, the following Corollary shows that following this protocol, every legitimate agent can learn the trustworthiness of all agents in the network, including their in- and out-neighbors, $\Nin$, $\Nout$, for all $i\in\LL$.
\begin{corollary}[Learning the Trustworthiness of All Agents] \label{cor:learning_all_agents}
All legitimate agents $i\in \LL$ can learn the trustworthiness of all agents in the network correctly. That is, there exists a finite time $T_{max}$ such that for all $t\geq T_{max}$ and for all $q\in \V$, $o_{iq}(t) \geq 1/2$ if $q\in \LL$ and $o_{iq}(t) < 1/2$ if $q\in \MM$ almost surely.
\end{corollary}
\begin{proof}
     Let $i$ be a legitimate agent. Let $q$ be an arbitrary agent in the system. Then, by Theorem \ref{thm:as_conv}, there exist a time $T_q$ almost surely such that for all $t\geq T_q$, $o_{iq}(t)>1/2$ if $q\in\LL$ and $o_{iq}(t)<1/2$ otherwise. Then we can choose $T_{max} = \max_{q \in \V} T_q$.
\end{proof}

\subsection{Finite-Time Analysis}
In this part, we analyze the finite-time behavior of our algorithm. The algorithm behaves in a deterministic way once we reach $T_f$, the time at which all legitimate agents learn the trustworthiness of their in-neighbors. Therefore, our first goal is to characterize the time it takes to learn all the agents after reaching $T_f$, i.e., the time interval between $T_f$ and $T_{\max}$ defined in \Cref{cor:learning_all_agents}. Using this, we derive upper bounds for the probabilities of reaching $T_{\max}$ at any time step $t$ and not reaching it before $t$. 

The characterization of $T_f$ follows from the previous work \cite{yemini2022resilience}. Let $D_{\mL}$ be the total number of legitimate in-neighbors in the system, i.e., $D_{\mL} \triangleq \sum_{i\in \mL} |\Nin \cap \mL|.$ Similarly, let $D_{\mM}$ be the total number of malicious in-neighbors in the system, i.e., $D_{\mM} \triangleq \sum_{i\in \mM} |\Nin \cap \mM|.$ We have the following result that concerning the probability of reaching $T_f$ at some time $k.$
\begin{lemma}[Lemma 2, \cite{yemini2022resilience}]
    \label{lemma:Tf_prob}
    Let $i \in \mL$ be a legitimate agent. We define $E_{\mL} \triangleq \mathbb{E}[\aij(t)]-1/2$ for $j \in \mL$ and $E_{\mM} \triangleq \mathbb{E}[\aij(t)]-1/2$ for $j \in \mM$. Define 
    $$p_c(t) \triangleq D_{\mL}\exp(-2tE_{\mL}^2)+D_{\mL}\exp(-2tE_{\mM}^2),$$
    $$p_e(t) \triangleq D_{\mL}\frac{\exp(-2tE_{\mL}^2)}{1-\exp(-2E_{\mL}^2)}+D_{\mM}\frac{\exp(-2tE_{\mM}^2)}{1-\exp(-2E_{\mM}^2)}.$$
    Then, we have the following upper bounds for all $t\geq 0$,
    \begin{align}
        \Pr(T_f=t) &\leq \min \{p_c(t-1),1\}, \\
        \Pr(T_f>t-1) &\leq \min \{p_e(t-1),1\}.
    \end{align}
\end{lemma}
 This lemma shows that probability of reaching $T_f$,  decreases exponentially with $t$ and is upper bounded by $p_c(t)$. Similarly, the probability of not reaching $T_f$ before some time $t$ also decreases exponentially with $t$ and is upper bounded by $p_e(t)$. We will use these upper bounds to derive similar relationships for reaching $T_{\max}$. Before doing so, we need to establish the relationship between $T_{\max}$ to $T_f$. With the next result, we show a bound on the time it takes to reach $T_{\max}$ after reaching $T_f.$ The bound is derived by exploiting the equivalence between opinion dynamics in our system and absorbing Markov chains.
 \begin{lemma}
 \label{lemma:Tf_Tmax_relation}
     Let $\deg_{\max} \triangleq~\max_{i\in \mL} |\Nin|$ be the maximum in-degree of any legitimate agent in the graph $\mG$, and let $l_{\mG}$ denote the diameter of the graph $\mG.$  Then, we have
     \begin{align}
         T_{\max}-T_f \leq h \cdot l_{\mG},
     \end{align}
     where $h \triangleq 1/\log_2 \frac{1}{1-(1/\deg_{\max})^{l_{\mG}}}.$
 \end{lemma}
 \begin{proof}
     If all agents are directly observing each other, then we have $T_{\max}=T_f$ trivially. Now, let $q\in \mV$ be an arbitrary agent in the system such that the set of legitimate agents not observing $q$ is non-empty, i.e.,  $\mC_q \neq \emptyset.$ Define the error in agents' opinions as $\errq(t)=\ol(t)-\ind.$ Notice that if the error in one agent's opinion is less than $1/2$, then, that agent classifies $q$ correctly. Therefore, if we have $\infnorm{\errq(t)}<1/2$, all agents classify $q$ correctly at time $t$. By \Cref{eq:delta_power} and properties of the matrix norm, we have
     \begin{align}
     \label{eq:err_vector_finite_bound}
         \infnorm{\errq(t)} &\leq \infnorm{\overline{W}_{\mC_q}^{t-T_f}} \infnorm{\errq(T_f)} \\
                            &\leq \infnorm{\overline{W}_{\mC_q}^{t-T_f}},
     \end{align}
    where the last step follows from the fact that the opinions can only take values in range $[0,1].$ Now, if we find a $t'$ such that $\infnorm{\overline{W}_{\mC_q}^{t'-T_f}}<1/2$, then, all the legitimate agents classify $q$ correctly after reaching $t'.$ To find $t'$, we define a Markov chain with the transition matrix $Q$. The states in the chain correspond to the agents in $\mC_q$ and the transition probability from agent $i \in \mC_q$ to $j \in \mC_q$ denoted by $[Q]_{ij}$ is given by $[Q]_{ij}= [\overline{W}_{\mC_q}]_{ij}$. This results in a graph where the corresponding edges in $\mG$ are flipped. 
    However, recall that $\overline{W}_{\mC_q}$ is a substochastic matrix. Therefore, we define an auxiliary state $a$ with transition probabilities $[Q]_{ia}=1-\sum_{j\in \mC_q} [\overline{W}_{\mC_q}]_{ij}$ and $[Q]_{aa} = 1$. Let $v$ be an $|\mC_q|\times 1$ vector with $v_i= [Q]_{ia}.$ Then, we have
    $$
        Q=\begin{bmatrix}\overline{W}_{\mC_q} & v \\
                         \boldsymbol{0}   & 1\end{bmatrix}, \text{ and } Q^t=\begin{bmatrix}\overline{W}_{\mC_q}^t & (\sum_{s=0}^{t-1} \overline{W}_{\mC_q}^s) v \\
                         \boldsymbol{0}   & 1\end{bmatrix}.
    $$
    where $\boldsymbol{0}$ is an $|\mC_q|\times 1$ vector of zeros. Note that since $\overline{W}_{\mC_q}$ is convergent, this transition matrix defines an absorbing Markov chain \cite{seneta2006non}. Notice that $\infnorm{\overline{W}_{\mC_q}^t} = 1 - \min_{i\in \mC_q} [(\sum_{s=0}^{t-1} \overline{W}_{\mC_q}^s) v]_i.$ The term $[(\sum_{s=0}^{t-1} \overline{W}_{\mC_q}^s) v]_i$ corresponds to the probability of going from the state $i$ to $a$ in $t$ steps. The minimum non-zero transition probability between any two agents is given by $1/\deg_{\max}$. Moreover, we know that there is a path of length at most $\widehat{con}\overline{W}_{\mC_q}$ from an arbitrary agent $i\in \mC_q$ to an agent $j$ with $\sum_{l\in \mC_q} [\overline{W}_{\mC_q}]_{jl}<1$ by \Cref{thm:W_C_weakly_chained}. By construction, this agent $j$ has an edge to agent $a$ with transition probability at least $1/\deg_{\max}.$ Therefore, for $t = \widehat{con}\overline{W}_{\mC_q}+1$, we have $[Q^{t}]_{ia}\geq (1/\deg_{\max})^{\widehat{con}\overline{W}_{\mC_q}+1}$ for any $i\in\mC_q.$ 
    
    Hence, we have
    \begin{align*}    \infnorm{\overline{W}_{\mC_q}^{\widehat{con}\overline{W}_{\mC_q}+1}} &\leq1- (1/\deg_{\max})^{\widehat{con}\overline{W}_{\mC_q}+1}, \\
        \infnorm{\overline{W}_{\mC_q}^{h_q(\widehat{con}\overline{W}_{\mC_q}+1)}} &\leq \infnorm{\overline{W}_{\mC_q}^{\widehat{con}\overline{W}_{\mC_q}+1}}^{h_q}\\
        &\leq (1-(1/\deg_{\max})^{(\widehat{con}\overline{W}_{\mC_q}+1)})^{h_{q}}. \\
    \end{align*}
    Then, for $h_q>1/\log_2 \frac{1}{1-(1/\deg_{\max})^{\widehat{con}\overline{W}_{\mC_q}+1}}$ we have $\infnorm{\overline{W}_{\mC_q}^{h_q(\widehat{con}\overline{W}_{\mC_q}+1)}}<1/2$. So, after $t'=h_q(\widehat{con}\overline{W}_{\mC_q}+1)+T_f$, all the legitimate agents classify $q$ correctly. We have $\max_{q\in \mV}\widehat{con}\overline{W}_{\mC_q}+1 \leq l_{\mG}$ where $l_{\mG}$ is the diameter of the graph. Let $h$ be $\max_{q\in \mV} h_q.$ Then, for all $t'\geq h\cdot l_{\mG}+T_f$ and for any $q\in \mV$, we have $\infnorm{\overline{W}_{\mC_q}^{t'}} < 1/2.$ Therefore, by the definition of $T_{\max}$ we get $T_{\max}-T_f \leq h\cdot l_{\mG}.$ 
 \end{proof}
 
 The following proposition gives us the desired probabilistic characterization of $T_{\max}.$
 \begin{proposition}
     Define $\Delta = h \cdot l_{\mG}+1$, where $h = 1/\log_2 \frac{1}{1-(1/\deg_{\max})^{l_{\mG}}}.$ Then, we have for all $t\geq 0,$
     \begin{align}
        \Pr(T_{\max}=t) &\leq \min \{p_c(t-\Delta),1\}, \text{ and,} \\
        \Pr(T_{\max}>t-1) &\leq \min \{p_e(t-\Delta),1\}.
    \end{align}
 \end{proposition}
\begin{proof}
    Define a variable $\Delta_{\text{diff}} \triangleq T_{\max}-T_f+1$. Notice that we have $\Delta_{\text{diff}} \leq \Delta = h \cdot l_{\mG}+1.$ Using \Cref{lemma:Tf_prob}, we get
    \begin{align*}
        \Pr(T_{\max}=t) &= \Pr(T_f=t-\Delta_{\text{diff}})\\
                    &\leq \min \{p_c(t-\Delta_{\text{diff}}),1\} \\
                    &\leq \min \{p_c(t-\Delta),1\},
    \end{align*}
    where in the last step we used the fact that $p_c$ is monotonically decreasing. The proof for $\Pr(T_{\max}>t-1)$ follows similarly.
\end{proof}

This proposition demonstrates that for sufficiently large $t$, the probability of reaching $T_{\max}$ at time $t$ and not reaching it before $t$ both decreases exponentially similar to $T_f$ (\Cref{lemma:Tf_prob}). Additionally, the upper bound on the time required for this exponential decrease is determined by the diameter and maximum degree of the communication graph. This proposition can be particularly useful for other algorithms that leverage our learning protocol to design resilient systems. For instance, given a desired level of confidence, one can determine a finite-time cut-off limit for running this learning protocol before starting another algorithm using only resulting trustworthy agents.

\section{Numerical Studies}
In this section, we evaluate the performance of the algorithm via numerical studies. We show that all legitimate agents can learn the the trustworthiness of all the other agents in the system using our algorithm in various network realizations, which supports our theoretical results. 

\textbf{Communication graph:} We generate the graph of legitimate agents, denoted with $\GL$, in two different ways to show that the protocol works with different graph structures. The first way is to use a cyclic graph. We choose this graph model because it is strongly connected by default and the contraction index for learning the trustworthiness of any legitimate agent in the system is $|\LL|-2$, which grows linearly with the number of legitimate agents. The second way is to generate a random graph using Erdős–Rényi model where each edge in the graph is either included or not with probability $p$ \cite{erdHos1960evolution}. We choose $p=\frac{2\log{|\LL|}}{|\LL|}$ to have a high probability of generating a strongly connected graph \cite{graham2008note}, and in the case where the generated graph is not strongly connected, we repeat the process to ensure satisfaction of Assumption \ref{assumption:graph-connectivity}.1. The graphs generated using this model is likely to have a better connectivity and a lower contraction index for learning any legitimate agent compared to the cyclic graphs. In this way the cyclic graph represents the most difficult case where legitimate information takes longest to circulate throughout the network. After generating the graph of legitimate agents, we randomly add malicious agents to the system.

\textbf{Malicious agents:} We assume that all the malicious agents in the system are omnipotent in that they know the trustworthiness of every other agent in the system, this represents a strong attack. The malicious agents do not follow the update rules and they always send the opposite of the true trustworthiness information to other agents, i.e. they assign $1$ to all malicious agents and assign $0$ to all legitimate agents in the trust vector they share. Since the malicious agents do not follow the learning protocol, we do not explicitly model the communication between the malicious agents.

\textbf{Trust observations:} Following the previous work \cite{ourTRO}, we model the trust observations $\aij(t)$ as follows: At each time step $t$ we sample $\aij(t)$ uniformly from the interval $[0.35, 0.75]$ if $j\in \LL$ and from $[0.25, 0.65]$ if $j\in \MM$. This way, $\mathbb{E}[\aij(t)]=0.55$ if $j$ is a legitimate agent and $\mathbb{E}[\aij(t)]=0.45$ otherwise. With this setup, Assumption \ref{assumption:trust-observations} is satisfied.

\textbf{Metrics:} We evaluate the model performance based on two different metrics. The first one is mean squared error (MSE) where we calculate the mean squared error between the true trust vector and trust vector of legitimate agents and take the average across all legitimate agents. Following from Theorem \ref{thm:as_conv}, the MSE should converge to $0$. The second metric we define is $\hat{T}_{max}$, which is a proxy for $T_{max}$ defined in Corollary \ref{cor:learning_all_agents}. If all legitimate agents classify every other agent correctly for $N$ number of time steps after time $t$, where $N$ is the total number of agents in the system, we assign $\hat{T}_{max}=t$ and stop the experiment, assuming that no further classification error would occur since $N$ is large enough for information to propagate through the whole network.

\subsection{Results}
Here, we present the results for three different setups with $|\LL|\in \{20,40,80\}$ and $|\MM|= 1.5 \times |\LL|$. For each $|\LL|$, we generate the network of legitimate agents in two different ways: using a cyclic graph, and an Erdős–Rényi graph over legitimate agents. Then we add the malicious agents randomly, and we track the MSE for both networks. Examples of these graph topologies can be seen in Figure \ref{fig:topology_examples} . The results are presented in Figure \ref{fig:num_res_three_setups}. It can be seen that both MSE and maximum error converges to $0$ in all setups.
\begin{figure}[t]
    \centering
    \subfigure[Cyclic $\GL$]{\includegraphics[width=0.33\textwidth]{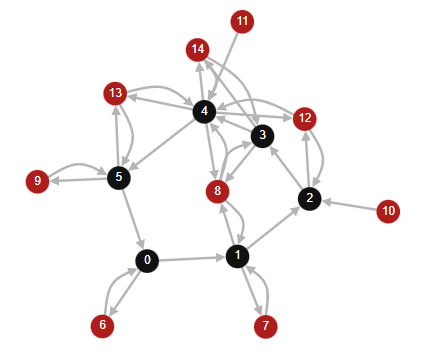}}
    \hspace{15mm}
    \subfigure[Erdős–Rényi $\GL$]{\includegraphics[width=0.33\textwidth]{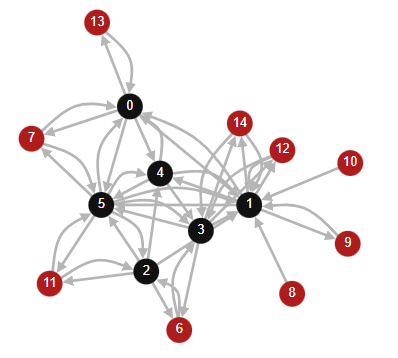}}
    \caption{Example graph topologies with $|\LL|=6$, $|\MM|=9$ nodes.}
    \label{fig:topology_examples}
\end{figure}

\begin{figure}
    \centering
    \subfigure[$|\LL|=20$, $|\MM|=30$]{\includegraphics[width=0.32\textwidth]{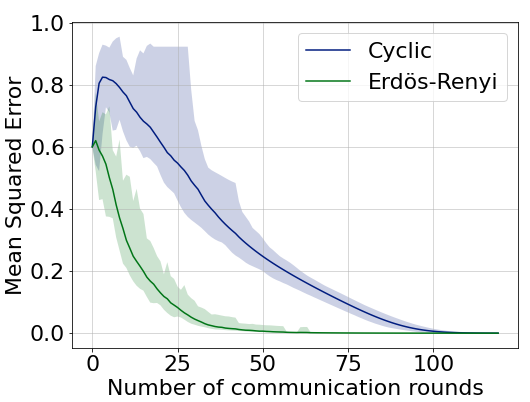}}
    \subfigure[$|\LL|=40$, $|\MM|=60$]{\includegraphics[width=0.32\textwidth]{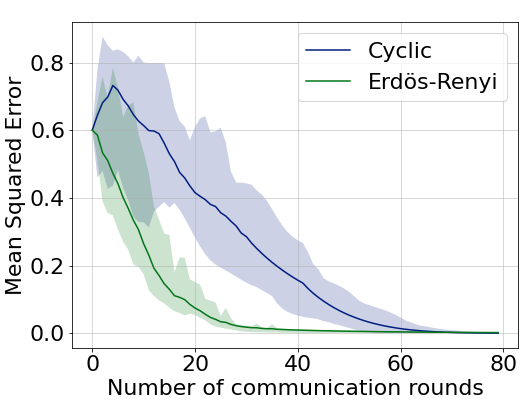}}
    \subfigure[$|\LL|=80$, $|\MM|=120$]{\includegraphics[width=0.32\textwidth]{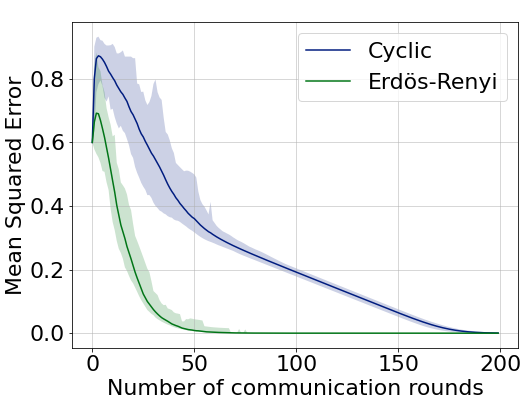}}
    \caption{Convergence plots for three different cases where the number of malicious agents are chosen as $|\MM|=1.5\times|\LL|$. Solid lines represents the MSE. The shaded areas show the range of error among the legitimate agents. We see that in all cases, the MSE converges to zero eventually as predicted by our theory. Since the malicious agents can have influence on the other nodes in the beginning, we observe an increase in the error at first. This effect is higher in cyclic graphs since the information takes longer to propagate. Moreover, as we increase the size of the graph for cyclic graphs, the convergence time also increases. On the other hand, since Erdős–Rényi graph has good connectivity in all cases, the convergence time is not as sensitive to the graph size compared to the cyclic graphs.}
    \label{fig:num_res_three_setups}
\end{figure}

For each setup, we present the maximum contraction index, denoted by $\widehat{con}_{max}$  and $\hat{T}_{max}$ in Table \ref{table:num_res_three_setups}. We define the maximum contraction index as $\widehat{con}_{max}=\max_{q\in \V}\widehat{con}\overline{W}_{\mC_q}$, where $\widehat{con}\overline{W}_{\mC_q}$ is defined in \eqref{eq:index_of_contraction}.

\begin{table}[]
\centering
\begin{tabular}{|l|cc|cc|cc|}
\hline
\multirow{2}{*}{}    & \multicolumn{2}{c|}{\textbf{$|\LL|=20$, $|\MM|=30$}}                                                & \multicolumn{2}{c|}{\textbf{$|\LL|=40$, $|\MM|=60$}}                                                & \multicolumn{2}{c|}{\textbf{$|\LL|=80$, $|\MM|=120$}}                                               \\ \cline{2-7} 
                     & \multicolumn{1}{l|}{\textbf{$\hat{T}_{max}$}} & \multicolumn{1}{l|}{\textbf{$\widehat{con}_{max}$}} & \multicolumn{1}{l|}{\textbf{$\hat{T}_{max}$}} & \multicolumn{1}{l|}{\textbf{$\widehat{con}_{max}$}} & \multicolumn{1}{l|}{\textbf{$\hat{T}_{max}$}} & \multicolumn{1}{l|}{\textbf{$\widehat{con}_{max}$}} \\ \hline
\textbf{Cyclic}      & \multicolumn{1}{c|}{66}                       & 19                                                  & \multicolumn{1}{c|}{109}                      & 38                                                  & \multicolumn{1}{c|}{192}                      & 78                                                  \\ \hline
\textbf{Erdős–Rényi} & \multicolumn{1}{c|}{49}                       & 3                                                   & \multicolumn{1}{c|}{64}                       & 3                                                   & \multicolumn{1}{c|}{76}                       & 3                                                   \\ \hline
\end{tabular}
\caption{This table shows $\hat{T}_{max}$ and $\widehat{con}_{max}$ for 8 different setups. We can see that a higher $\widehat{con}_{max}$ usually corresponds to a higher $\hat{T}_{max}$. This correlation is intuitive since $\widehat{con}_{max}$ is an indicator of how long it takes for information to propagate from observing agents to non-observing agents. Since the graph topology dictates $\widehat{con}_{max}$, we observe a higher $\hat{T}_{max}$ in Cyclic graphs compared to Erdős–Rényi graphs.}
\label{table:num_res_three_setups}
\end{table}
\subsection{The Effect of Malicious Agents}
In this part, we investigate the effect of malicious agents in the system to the learning protocol. We use the Erdős–Rényi graph setup from the previous part with $40$ legitimate agents. We look into two cases: First, we fix the number of malicious agents in the system to $60$ and we change the likelihood that the malicious agents make a connection with a legitimate agent. Then, we fix the probability of making a connection to $0.2$ and increase the number of malicious agents in the system. The MSE graphs are shown in Fig \ref{fig:malicious_effect}
\begin{figure}[!b]
    \centering
    \subfigure[]{\includegraphics[width=0.45\textwidth]{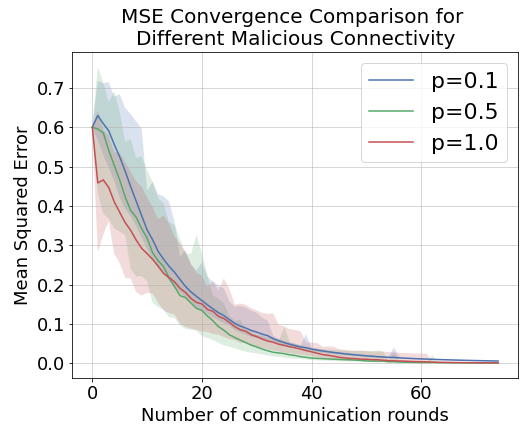}}
    \subfigure[]{\includegraphics[width=0.45\textwidth]{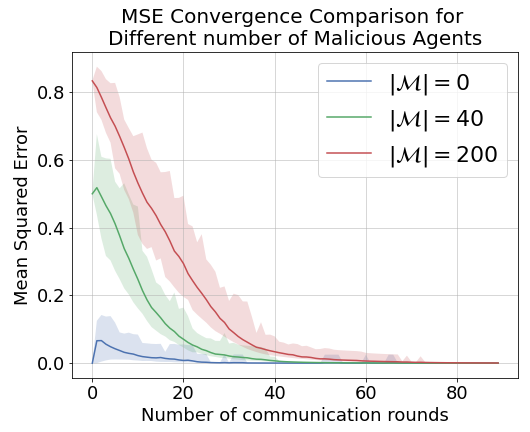}}
    \caption{(a) The effect of increasing the expected number of connections that each malicious agent has. Here, $p$ denotes the probability that the edge $(i,m)$ is present in the system, meaning that $m\in\MM$ is an in-neighbor of $i\in \LL$. 
    As malicious agents are being observed by more agents, their early effect in the network decreases since all of the directly observing agents learn their trustworthiness using their own observations, without waiting for the information to propagate from the other learning agents. (b) The effect of increasing the number of malicious agents in the system.  As we increase the number of malicious agents in the system, their effect on the legitimate agents' opinions also increase. However, the opinions of the legitimate agents still converge to the correct values eventually demonstrating agreement with our main result in Theorems \ref{thm:as_conv} and \ref{thm:convMean}.}
    \label{fig:malicious_effect}
\end{figure}

\subsection{Necessity of Assumption \ref{assumption:graph-connectivity}.2}
In our experiments, we empirically demonstrate the necessity of Assumption \ref{assumption:graph-connectivity}.2 for the learning protocol to work. We generated a simple example with two legitimate agents and two malicious agents, which can be seen in Figure \ref{fig:assumption_1_2}. In this example, the malicious agent $1$ is not an in-neighbor of either of the legitimate nodes while the malicious agent $2$ is an in-neighbor of both of the legitimate agents. With this setup, both of the legitimate agents failed to learn the identity of the agent $m_1$ as expected while learning all of the other agents successfully in ten different trials. This is because before the legitimate agents learn the trustworthiness of the malicious agent $2$, the malicious agent $2$ changes the opinions of the legitimate agents about the malicious agent $1$. After that, since none of the legitimate agents is directly observing the malicious agent $1$, their opinion does not change.

\begin{figure}[]
    \centering
    \includegraphics[scale=0.5]{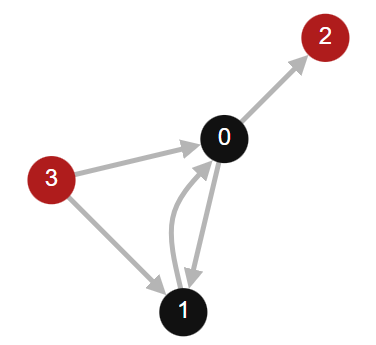}
    \caption{ An example network topology where the learning is not guaranteed and the Assumption \ref{assumption:graph-connectivity}.2 is violated. Since the agent 2 is not an in-neighbor of either of the legitimate agents, they fail to learn the trustworthiness of this agent.}
    \label{fig:assumption_1_2}
\end{figure}

\subsection{Aggregate Results}
We present numerical results over multiple trials in this section, each representing a random instantiation of the graph topology and stochastic observations of trust drawn from the distribution described at the beginning of this section. We fix the number of legitimate agents to $40$ and the number of malicious agents to $60$. Then, we run 500 trials for both cyclic graph and Erdős–Rényi setups. For each trial, we run the protocol for 1000 communication rounds and record the $\hat{T}_{max}$. The resulting statistics of $\hat{T}_{max}$ are shown in Table \ref{table:aggregate}. 

\begin{table}[]
\centering
\begin{tabular}{|l|l|l|l|l|}
\hline
                     & \multicolumn{1}{l|}{\textbf{min}} & \multicolumn{1}{l|}{\textbf{max}} & \multicolumn{1}{l|}{\textbf{mean}} & \multicolumn{1}{l|}{\textbf{std}} \\ \hline
\textbf{Cyclic}      & 97                                & 175                               & 112.8                              & 8.85                              \\ \hline
\textbf{Erdős–Rényi} & 36                                & 164                               & 58.2                               & 14.8                              \\ \hline
\end{tabular}
\caption{This table shows the summary statistics of $\hat{T}_{max}$ calculated over 500 random trials. In all trials, we observe a finite time $\hat{T}_{max}$ where no classification errors occur thereafter as predicted by the theory. Since the connections between legitimate agents in Erdős–Rényi graph is random, a higher variation in $\hat{T}_{max}$ is observed.}
\label{table:aggregate}
\end{table}

\section{Conclusion}

This paper presents a protocol for learning which agents to trust, and the accompanying analysis, for directed multiagent graphs with stochastic observations of trust.  Here, the directed nature of the graph presents an important challenge where the out-neighbors of a node cannot directly observe or receive information from it; this leads to a learning dynamic that makes accurate assessment of malicious agents in the network particularly elusive. The learning protocol developed herein specifically addresses this challenge of learning trust in directed graphs and constitutes the novelty of this paper. Since directed graphs often arise in practical multiagent systems due to heterogeneity in sensing and communication, we believe that the learning protocol and theory presented here can support many optimization, estimation, and learning tasks for general multiagent systems.

\acks{
The authors gratefully acknowledge partial support through NSF awards CNS 2147641, CNS-2147694, and AFOSR grant number FA9550-22-1-0223. The authors would like to thank Ali Taherinassaj for their helpful discussions.
}
\bibliographystyle{plain}
\bibliography{references.bib}
\end{document}